\documentclass{article}
\usepackage[a4paper, total={6in, 9in}]{geometry}
\usepackage{varioref}
\usepackage{prettyref}
\usepackage{url}
\usepackage{xcolor}
\usepackage{color}
\usepackage{graphicx}
\usepackage{titlesec}
\usepackage{indentfirst}
\usepackage{mathtools}
\usepackage{nomencl}
\usepackage{amsmath}
\usepackage{amsthm}
\usepackage{amssymb}  % For more math
\usepackage{lipsum}
\usepackage{tikz}
\usepackage{float}
\usepackage{mdframed} % algorithm box
\let\counterwithin\relax
\usepackage{chngcntr}
\counterwithin{figure}{section}

\usepackage[scr]{rsfso}

\usetikzlibrary{matrix,decorations.pathreplacing}
\newtheorem{theorem}{Theorem}

\newtheorem{assumption}{Assumption}

\usepackage{mathtools} % Bonus

\begin{document}

% -------------------------------------
\title{Consistency Analysis of the Closed-loop SRIVC Estimator}%
\author{Siqi Pan, James S. Welsh, Rodrigo A. Gonz\'alez and Cristian R. Rojas}%
\date{}
\maketitle
% ------------------------------------

\section{Problem Formulation}

Consider an asymptotically stable, linear, continuous-time closed-loop system
\begin{equation}
\mathcal{S}\colon \begin{cases}
               y(t) &= G^*(p)u(t) + v(t)	\\
	     u(t) &= r(t) - C(p)y(t),
            \end{cases}
\end{equation}
where $G^*(p)=B^*(p)/A^*(p)$ denotes the continuous-time system, $C(p)$ denotes the continuous-time controller, 
$u(t)$ is the plant input and $y(t)$ is the plant output. The external reference signal, $r(t)$, is given by
$r(t) = r_1(t) + C(p)r_2(t)$. For simplicity of the analysis, we assume that the additive noise on the output is discrete-time 
in nature and will have a ZOH or FOH intersample behaviour when entering the closed-loop system a
seen in Figure~\ref{fig:closed-loop1}.

\begin{figure} [h]
\begin{center}
\includegraphics[width = 14cm]{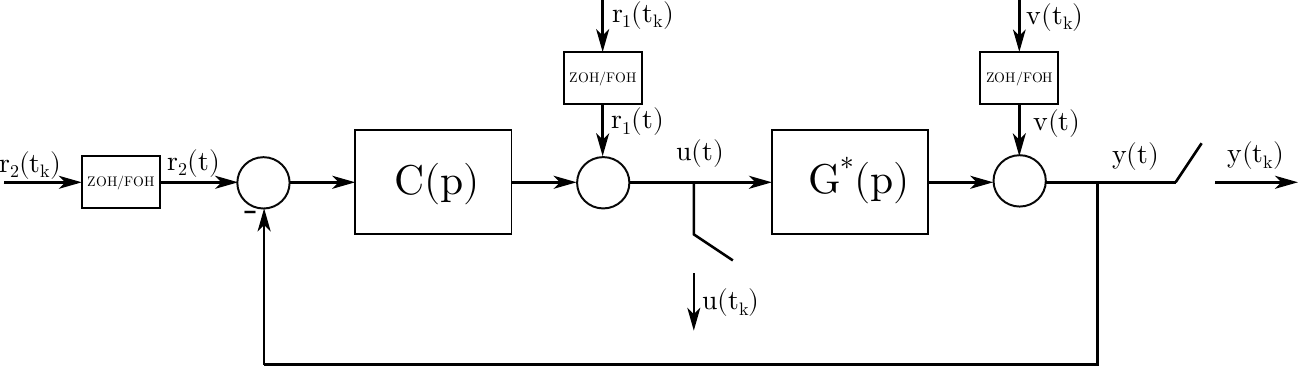}
\caption{Closed-loop configuration in the existing literature.}
\label{fig:closed-loop1}
\end{center}
\end{figure}

The continuous-time system identification problem is to find an estimate of $G^*(p)$ from sampled data 
$\{r(t_k)\}_{k=1}^N$, $\{u(t_k)\}_{k=1}^N$ and $\{y(t_k)\}_{k=1}^N$. The plant input and output 
after they get sampled can be expressed through sensitivity functions as
\begin{align}
	u(t_k) &= \frac{1}{1+G^*(p)C(p)}r_1(t_k) 
		+ \frac{C(p)}{1+G^*(p)C(p)}r_2(t_k) - \frac{C(p)}{1+G^*(p)C(p)}v(t_k)	\label{eq:plant_input}\\
	y(t_k) &= \frac{G^*(p)}{1+G^*(p)C(p)}r_1(t_k)+ \frac{G^*(p)C(p)}{1+G^*(p)C(p)}r_2(t_k) 
	+ \frac{1}{1+G^*(p)C(p)}v(t_k). \label{eq:plant_output}
\end{align}
Note that, according to the closed-loop configuration in Figure~\ref{fig:closed-loop1}, each sensitivity functions 
in~\eqref{eq:plant_input} and~\eqref{eq:plant_output} need to be discretised according 
to the intersample behaviours of the reference signals and the noise. 
Also note that the external reference will either be $r_1(t_k)$ or $r_2(t_k)$ but not both simultaneously.

The sampled versions of the noise-free plant input and output are expressed as
\begin{align}
	u^*(t_k) &= \frac{1}{1+C(p)G^*(p)}r_1(t_k) 
	+ \frac{C(p)}{1+C(p)G^*(p)}r_2(t_k) \label{eq:plant_input_free}\\
	x^*(t_k) &= \frac{G^*(p)}{1+C(p)G^*(p)}r_1(t_k)+ \frac{G^*(p)C(p)}{1+C(p)G^*(p)}r_2(t_k).
	\label{eq:plant_output_free}
\end{align}
An approximate version of these noise-free signals will be used to construct the instrument vectors in the 
CLSRIVC estimator since these signals are not available for measurements.
% Note that these signals can't be labelled on the block diagram due to the noise in the feedback loop

%%%%%%%%%%%%%%%
Assume that the controller $C(p)$ is known, then the $(j+1)$-th iteration of the CLSRIVC estimator is given by
\begin{equation}	\label{eq:clsrivc}
	\boldsymbol\theta_{j+1} = \left[\frac{1}{N}\sum_{k=1}^N \hat{\boldsymbol\varphi}_f(t_k,\boldsymbol\theta_j)	
	\boldsymbol\varphi^\top_f(t_k,\boldsymbol\theta_j)\right]^{-1}	
	\left[\frac{1}{N}\sum_{k=1}^N \hat{\boldsymbol\varphi}_f(t_k,\boldsymbol\theta_j) y_f(t_k,\boldsymbol\theta_j)\right],	
\end{equation}
where the filtered regressor vector and the filtered output are given by
\begin{equation}	\label{eq:cl_reg}
	\boldsymbol\varphi_f(t_k,\boldsymbol\theta_j) = \left[\begin{array}{cccccc}
	-\dfrac{p^{n}}{A_j(p)}y(t_k) & \dots & -\dfrac{p}{A_j(p)}y(t_k) & \dfrac{p^m}{A_j(p)}u(t_k) 
	& \dots & \dfrac{1}{A_j(p)}u(t_k)
	\end{array}\right]^\top,
\end{equation}
and
\begin{equation*}
	y_f(t_k,\boldsymbol\theta_j) = \frac{1}{A_j(p)}y(t_k)
\end{equation*}
respectively. The filtered instrument vector needs to be chosen such that it is not correlated with the additive 
noise on the output. Hence, directly constructing the filtered instrument vector using the plant input $u(t_k)$ as in the 
open-loop case is not a viable option in the closed-loop case. Instead, the noise-free plant input and 
output given in~\eqref{eq:plant_input_free} and~\eqref{eq:plant_output_free} 
respectively can be approximated using the sensitivity functions constructed using the previous estimate 
of the plant, and the filtered instrument vector is then constructed using the filtered derivatives of the 
approximated input and output signals, i.e.
\begin{equation}	\label{eq:cl_ins}
	\hat{\boldsymbol\varphi}_f(t_k,\boldsymbol\theta_j) = \left[\begin{array}{cccccc}
	-\hat{x}_f^{(n)}(t_k) & \dots & -\hat{x}^{(1)}_f(t_k) & \hat{u}_f^{(m)}(t_k) & \dots & \hat{u}_f(t_k)
	\end{array}\right]^\top,
\end{equation}
where
\begin{align*}
	\hat{x}_f^{(i)} &= \frac{p^i G_j(p)}{A_j(p)(1+C(p)G_j(p))} r_1(t_k)	
	+ \frac{p^i G_j(p)C(p)}{A_j(p)(1+C(p)G_j(p))} r_2(t_k)\\
	\hat{u}_f^{(i)} &= \frac{p^i }{A_j(p)(1+C(p)G_j(p))} r_1(t_k)
	+ \frac{p^i C(p)}{A_j(p)(1+C(p)G_j(p))} r_2(t_k).
\end{align*}
% Note that we might not necessarily know the intersample behaviour of r(t_k) since it might be filtered by C(p)

Let the plant, the $j$-th iteration of the model and controller be parameterised as
\begin{equation}	\label{eq:plant}
	G^*(p) = \frac{B^*(p)}{A^*(p)} = \frac{b^*_0 p^{m^*} + \cdots +b^*_{m^*}}{a^*_1p^{n^*} + \cdots + a^*_{n^*}p + 1}
\end{equation}
with $n^*\geq m^*$, 
\begin{equation*}
	G_j(p) = \frac{B_j(p)}{A_j(p)} = \frac{b_0 p^{m} + \cdots +b_{m^*}}{a_1p^{n} + \cdots + a_{n}p + 1}
\end{equation*}
with $n \geq m$, and
\begin{equation}	\label{eq:controller}
	C(p) = \frac{F(p)}{L(p)}
\end{equation}
with numerator order $n_f$ and denominator order $n_l$.

%%%

\section{Theoretical Analysis}

\begin{assumption}	\label{A1}
The system, $\frac{B^*(p)}{A^*(p)}$, and controller, $\frac{F(p)}{L(p)}$, are asymptotically stable 
with $B^*(p)$ and $A^*(p)$ being coprime and $n^* \geq m^*$. 
The closed-loop system is also asymptotically stable, i.e. the zeros of $A^*(p)L(p)+B^*(p)F(p)$ 
have strictly negative real parts.
\end{assumption}

\begin{assumption}	\label{A2}
The external reference, $r(t_k)$, and disturbance, $v(t_s)$, are stationary and mutually independent 
for all $k$ and $s$.
\end{assumption}

\begin{assumption}	\label{A3}
The external reference, $r(t_k)$, is persistently exciting of order no less than 
$n+\max(n+n_l,m+n_f)+1$.
\end{assumption}

\begin{assumption}	\label{A4}
The zeros of $A_j(p)$ have strictly negative real parts, $n \geq m$, with $A_j(p)$ and $B_j(p)$ 
being coprime. The model of the closed-loop system is also asymptotically stable, i.e. the zeros 
of $A_j(p)L(p)+B_j(p)F(p)$ have strictly negative real parts.
\end{assumption}

\begin{assumption}	\label{A5}
The model order is known exactly, i.e. $n=n^*$ and $m=m^*$.
\end{assumption}

\begin{assumption}	\label{A6}
The sampling frequency is larger than twice of the largest imaginary part of the zeros of 
$A_j(p)(A^*(p)L(p)+B^*(p)F(p))$.
\end{assumption}

%%%
Assume that $r_1(t_k)$ is the external reference signal and $r_2(t_k)=0$ (note that the proof should be 
very similar for a non-zero $r_2(t_k)$ with some differences in the sensitivity functions). 
Then, the filtered regressor vector~\eqref{eq:cl_reg} can be written as
\begin{align*}
	\boldsymbol\varphi_f(t_k,\boldsymbol\theta_j) &= \Bigg[\begin{array}{ccc}
	-\dfrac{p^{n}}{A_j(p)}\left\{\dfrac{G^*(p)}{1+G^*(p)C(p)} r(t)\right\}_{t=t_k}
	& \dots
	& -\dfrac{p}{A_j(p)}\left\{\dfrac{G^*(p)}{1+G^*(p)C(p)} r(t)\right\}_{t=t_k}
	\end{array}\\
	&\begin{array}{ccc}
	 \dfrac{p^{m}}{A_j(p)}\left\{\dfrac{1}{1+G^*(p)C(p)} r(t)\right\}_{t=t_k}
	& \dots
	& \dfrac{1}{A_j(p)}\left\{\dfrac{1}{1+G^*(p)C(p)} r(t)\right\}_{t=t_k}
	\end{array}\Bigg]^\top
	+ \mathbf{v}_f(t_k),
\end{align*}
where $\mathbf{v}_f(t_k)$ is a vector containing the filtered versions of the additive noise, $v(t_k)$.

The filtered regressor vector can also be written as
\begin{equation}	\label{eq:phi_3parts}
	\boldsymbol\varphi_f(t_k,\boldsymbol\theta_j) = 
	\tilde{\boldsymbol\varphi}_f(t_k,\boldsymbol\theta_j) + \boldsymbol\Delta(t_k,\boldsymbol\theta_j)
	+ \mathbf{v}_f(t_k,\boldsymbol\theta_j),
\end{equation}
where
\begin{align}		\label{eq:phi_tilde}
	\tilde{\boldsymbol\varphi}_f(t_k,\boldsymbol\theta_j) &= \Big[\begin{array}{ccc}
	-\dfrac{p^{n}G^*(p)}{A_j(p)(1+G^*(p)C(p))} r(t_k)
	& \dots
	& -\dfrac{pG^*(p)}{A_j(p)(1+G^*(p)C(p))} r(t_k)
	\end{array} \notag\\
	&\hspace{2cm}\begin{array}{ccc}
	 \dfrac{p^{m}}{A_j(p)(1+G^*(p)C(p))} r(t_k)
	& \dots
	&\dfrac{1}{A_j(p)(1+G^*(p)C(p))} r(t_k)
	\end{array}\Big]^\top.
\end{align}

Now, substitute~\eqref{eq:plant} and~\eqref{eq:controller} into~\eqref{eq:phi_tilde}, and let
\begin{equation*}
	Q^*(p) := A^*(p)L(p) + B^*(p)F(p),
\end{equation*}
we can then express~\eqref{eq:phi_tilde} as
\begin{align}	\label{eq:phi_tilde2}
	\tilde{\boldsymbol\varphi}_f(t_k,\boldsymbol\theta_j) &= \left[\hspace{-0.1cm}\begin{array}{cccccc}
	-\dfrac{p^{n}B^*(p)L(p)}{A_j(p)Q^*(p)} r(t_k)\hspace{-0.2cm}
	& \dots \hspace{-0.2cm}
	& -\dfrac{pB^*(p)L(p)}{A_j(p)Q^*(p)} r(t_k)
	 &\dfrac{p^{m}A^*(p)L(p)}{A_j(p)Q^*(p)} r(t_k) \hspace{-0.2cm}
	& \dots \hspace{-0.2cm}
	&\dfrac{A^*(p)L(p)}{A_j(p)Q^*(p)} r(t_k)
	\end{array}\hspace{-0.1cm}\right]^\top\notag \\
	&= \mathbf{S}(-B^*,A^*) \frac{L(p)}{A_j(p)Q^*(p)}\mathbf{r}_d(t_k),
\end{align}
where $\mathbf{S}(-B^*,A^*)$ is a $(n+m+1)\times(n+m+1)$ Sylvester matrix constructed using the coefficients 
of $B^*(p)$ and $A^*(p)$. This matrix is non-singular under Assumption~\ref{A1}.
The vector $\mathbf{r}_d(t_k)$ is given by
\begin{equation}	\label{eq:rd}
	\mathbf{r}_d(t_k) = \left[\begin{array}{ccc}
	r^{(n+m)}(t_k) & \dots & r(t_k)
	\end{array}\right]^\top.
\end{equation}

The vector and $\boldsymbol\Delta(t_k,\boldsymbol\theta_j)$ contains the interpolation error of the filtered derivatives 
of the plant output, i.e.
\begin{equation}	\label{eq:Delta}
\boldsymbol\Delta_i(t_k,\boldsymbol\theta_j) = \begin{cases}
\dfrac{p^{n+1-i}B^*(p)L(p)}{A_j(p)Q^*(p)} r(t_k) 
- \dfrac{p^{n+1-i}}{A_j(p)}\left\{\dfrac{B^*(p)L(p)}{Q^*(p)}r(t)\right\}_{t=t_k}\; &\text{for} \;i=1,\dots, n,\\
0,\;&\text{otherwise}.
\end{cases}
\end{equation}

%%%

Similarly, the filtered instrument vector can be written as
\begin{equation}	\label{eq:phi_hat}
	\hat{\boldsymbol\varphi}_f(t_k,\boldsymbol\theta_j) = 
	\mathbf{S}(-B_j,A_j)\frac{L(p)}{A_j(p)Q_j(p)} \mathbf{r}_d(t_k)
\end{equation}
where $\mathbf{r}_d(t_k)$ is given in~\eqref{eq:rd} and
\begin{equation*}
	Q_j(p) := A_j(p)L(p)+B_j(p)F(p).
\end{equation*}

Now, as $N\rightarrow\infty$, the modified normal matrix of the CLSRIVC estimator can be expressed as
\begin{equation}	\label{eq:final_normal}
	\mathbb{E}\left\{\hat{\boldsymbol\varphi}_f(t_k,\boldsymbol\theta_j)
	\boldsymbol\varphi^\top_f(t_k,\boldsymbol\theta_j)\right\}
	= \underbrace{\mathbf{S}(-B_j,A_j)\boldsymbol\Phi\mathbf{S}^\top(-B^*,A^*)}_
	{\mathbb{E}\left\{\hat{\boldsymbol\varphi}_f(t_k,\boldsymbol\theta_j)
	\tilde{\boldsymbol\varphi}^\top_f(t_k,\boldsymbol\theta_j)\right\}}
	+ \mathbb{E}\left\{\hat{\boldsymbol\varphi}_f(t_k,\boldsymbol\theta_j) \boldsymbol\Delta^\top(t_k)\right\} 
	+  \mathbb{E}\left\{\hat{\boldsymbol\varphi}_f(t_k,\boldsymbol\theta_j) \mathbf{v}_f^\top(t_k)\right\}
\end{equation}
where
\begin{equation}	\label{eq:Phi}
	\boldsymbol\Phi := \mathbb{E}\left\{\frac{L(p)}{A_j(p)Q_j(p)} {\mathbf{r}_d}(t_k)
	\frac{L(p)}{A_j(p)Q^*(p)} {\mathbf{r}_d}^\top(t_k)\right\}.
\end{equation}

Next, we examine the consistency of the CLSRIVC estimator in Theorem~\ref{thm:clsrivc_consistency}.

\begin{theorem}\label{thm:clsrivc_consistency}
Consider the CLSRIVC estimator~\eqref{eq:clsrivc}. Suppose Assumptions~\ref{A1} --~\ref{A6} hold and 
the condition
\begin{equation*}
\left\|\mathbb{E}\left\{\hat{\boldsymbol{\varphi}}_f(t_k,\boldsymbol\theta_j)\boldsymbol{\Delta}^\top(t_k,\boldsymbol\theta_j)\right\}\right\|_2
< \sigma_{n+m+1}
\left( \mathbb{E}\left\{\hat{\boldsymbol{\varphi}}_f(t_k,\boldsymbol\theta_j)\tilde{\boldsymbol{\varphi}}_f^\top(t_k,\boldsymbol\theta_j)\right\}\right)
\end{equation*}
is satisfied, where $\tilde{\boldsymbol{\varphi}}_f(t_k,\boldsymbol\theta_j)$ and 
$\boldsymbol{\Delta}(t_k,\boldsymbol\theta_j)$ are given 
by~\eqref{eq:phi_tilde2} and~\eqref{eq:Delta}, respectively. 
Then for a ZOH or FOH external reference and a continuous-time controller, $C(p)$, 
the CLSRIVC estimator is generically inconsistent 
when only sampled data are available as measurements.
\end{theorem}

\begin{proof}

Consider the modified normal matrix in~\eqref{eq:final_normal}. It can be shown that 
\begin{equation*}
\mathbb{E}\left\{\hat{\boldsymbol\varphi}_f(t_k,\boldsymbol\theta_j) \mathbf{v}_f^\top(t_k)\right\} = \mathbf{0}
\end{equation*}
under Assumption~\ref{A2} by following the same procedure as the proof of Statement~1 of Theorem~1 in~\cite{Pan2020}.
Now,~\eqref{eq:final_normal} simplifies to
\begin{equation*}
	\mathbb{E}\left\{\hat{\boldsymbol\varphi}_f(t_k,\boldsymbol\theta_j)
	\boldsymbol\varphi^\top_f(t_k,\boldsymbol\theta_j)\right\}
	= \mathbf{S}(-B_j,A_j)\boldsymbol\Phi\mathbf{S}^\top(-B^*,A^*)
	+ \mathbb{E}\left\{\hat{\boldsymbol\varphi}_f(t_k,\boldsymbol\theta_j) \boldsymbol\Delta^\top(t_k)\right\}.
\end{equation*}

Let the set of parameters that describe $A_j(p)$ and $B_j(p)$ be defined as
$$\Omega  = \{a_1, \dots, a_n, b_0, \dots, b_m \colon A_j(p) \; \text{is a stable polynomial}\}.$$

Define $\boldsymbol\Phi^*$ as the matrix with $A_j(p)=A^*$ and $B_j(p)=B^*(p)$, i.e.
\begin{equation*}
	\boldsymbol\Phi^* := \mathbb{E}\left\{\frac{L(p)}{A^*(p)Q^*(p)} \mathbf{r}_d(t_k)
	\frac{L(p)}{A^*(p)Q^*(p)} \mathbf{r}_d^\top(t_k)\right\}.
\end{equation*}

The polynomial, $A^*(p)Q^*(p)$, has a degree of $n+\max(n+n_l,m+n_f)$. Then, by following the same procedure as the 
proof of Lemma~7 in~\cite{Pan2020}, $\boldsymbol\Phi^*$ can be shown to be positive definite for an input with a 
persistent excitation of $n+\max(n+n_l,m+n_f)+1$. 

Next, an arbitrary entry of $\boldsymbol\Phi$ in~\eqref{eq:Phi} can be written as
\begin{align*}
	\boldsymbol\Phi_{dl} &= \mathbb{E}\left\{\frac{p^{n+m+1-d} L(p)}{A_j(p)Q_j(p)}r(t_k)
	\frac{p^{n+m+1-l} L(p)}{A_j(p)Q^*(p)}r(t_k)\right\}\\
	&= \frac{1}{2\pi}\int_{-\pi}^{\pi}\frac{C_1(e^{i\omega})}{\tilde{A}_j(e^{i\omega})\tilde{Q}_j(e^{i\omega})}
	\frac{C_2(e^{-i\omega})}{\tilde{A}_j(e^{-i\omega})\tilde{Q}^*(e^{-i\omega})} dF_{r}(\omega).
\end{align*}
By following following the same procedure as the proof of Lemma~9 in~\cite{Pan2020}, it can be shown that 
the entries of $\boldsymbol\Phi$ are analytic functions of every parameter in $\Omega$ due to the fact that 
compositions of analytic functions are analytic.

The non-singularity of the modified normal matrix of the CLSRIVC estimator guarantees the existence of a solution 
to the CLSRIVC esitmator. Next, let's examine the converging point of the CLSRIVC estimator.

Provided that the CLSRIVC estimator~\eqref{eq:clsrivc} converges, the converging point must satisfy
\begin{equation}	\label{eq:converging_point}
	\bar{\boldsymbol\theta} = \left[\frac{1}{N}\sum_{k=1}^N \hat{\boldsymbol\varphi}_f(t_k,\bar{\boldsymbol\theta})	
	\boldsymbol\varphi^\top_f(t_k,\bar{\boldsymbol\theta})\right]^{-1}	
	\left[\frac{1}{N}\sum_{k=1}^N \hat{\boldsymbol\varphi}_f(t_k,\bar{\boldsymbol\theta}) y_f(t_k,\bar{\boldsymbol\theta})\right].	
\end{equation}

As $N\rightarrow\infty$,~\eqref{eq:converging_point} can be written as
\begin{equation*}
	\mathbb{E}\left\{\hat{\boldsymbol\varphi}_f(t_k,\bar{\boldsymbol\theta})	
	\boldsymbol\varphi^\top_f(t_k,\bar{\boldsymbol\theta})\right\}^{-1}	
	\mathbb{E}\left\{\hat{\boldsymbol\varphi}_f(t_k,\bar{\boldsymbol\theta}) 
	\left(y_f(t_k,\bar{\boldsymbol\theta}) - \boldsymbol\varphi^\top_f(t_k,\bar{\boldsymbol\theta})\bar{\boldsymbol\theta}\right)\right\} 
	= \mathbf{0}.
\end{equation*}

The non-singularity of the modified normal matrix implies that
\begin{align*}
	\mathbb{E}\left\{\hat{\boldsymbol\varphi}_f(t_k,\bar{\boldsymbol\theta}) 
	\left(y_f(t_k,\bar{\boldsymbol\theta}) - \boldsymbol\varphi^\top_f(t_k,\bar{\boldsymbol\theta})\bar{\boldsymbol\theta}\right)\right\} 
	&= \mathbf{0}	\\
	\mathbb{E}\left\{\hat{\boldsymbol\varphi}_f(t_k,\bar{\boldsymbol\theta}) 
	\left(\frac{1}{\bar{A}(p)}y(t_k) + \frac{\bar{a}_1p^n+\cdots+\bar{a}_np}{\bar{A}(p)}y(t_k) 
	-\frac{\bar{B}(p)}{\bar{A}(p)}u(t_k) \right)\right\} 
	&= \mathbf{0}\\
	\mathbb{E}\left\{\hat{\boldsymbol\varphi}_f(t_k,\bar{\boldsymbol\theta}) 
	\left(y(t_k) -\frac{\bar{B}(p)}{\bar{A}(p)}u(t_k) \right)\right\} 
	&= \mathbf{0}\\
	\mathbb{E}\left\{\hat{\boldsymbol\varphi}_f(t_k,\bar{\boldsymbol\theta}) 
	\left(\left\{\frac{B^*(p)}{A^*(p)}u(t)\right\}_{t=t_k} -\frac{\bar{B}(p)}{\bar{A}(p)}u(t_k) \right)\right\} 
	+ \mathbb{E}\left\{\hat{\boldsymbol\varphi}_f(t_k,\bar{\boldsymbol\theta}) v(t_k)\right\} 
	&= \mathbf{0}.
\end{align*}

Since the additive noise on the output, $v(t_k)$, is independent of the input under Assumption~\ref{A2}, 
by following the same procedure as the proof of Statement~3 in~\cite{Pan2020},
\begin{equation*}
\mathbb{E}\left\{\hat{\boldsymbol\varphi}_f(t_k,\bar{\boldsymbol\theta}) v(t_k)\right\} = \mathbf{0}.
\end{equation*}

Hence, the converging point of the CLSRIVC estimator must satisfy
\begin{align}	\label{eq:epsilon}
	\mathbb{E}\left\{\hat{\boldsymbol\varphi}_f(t_k,\bar{\boldsymbol\theta}) 
	\left(\left\{\frac{B^*(p)}{A^*(p)}u(t)\right\}_{t=t_k} -\frac{\bar{B}(p)}{\bar{A}(p)}u(t_k) \right)\right\} 
	&= \mathbf{0}\\
	\mathbb{E}\left\{\hat{\boldsymbol\varphi}_f(t_k,\bar{\boldsymbol\theta}) 
	\underbrace{\left(\left\{\frac{B^*(p)}{A^*(p)(1+G^*(p)C(p))}r(t)\right\}_{t=t_k} 
	-\frac{\bar{B}(p)}{\bar{A}(p)}\left\{\frac{1}{1+G^*(p)C(p)}r(t)\right\}_{t=t_k} \right)}_{=\varepsilon(t_k,\bar{\boldsymbol\theta})}\right\} 
	&= \mathbf{0}.
\end{align}

Now, introduce an input-dependent term
\begin{equation*}
	\varepsilon_u(t_k) = \left\{\frac{\bar{B}(p)}{\bar{A}(p)(1+G^*(p)C(p))}r(t)\right\}_{t=t_k}
	-\frac{\bar{B}(p)}{\bar{A}(p)}\left\{\frac{1}{1+G^*(p)C(p)}r(t)\right\}_{t=t_k}.
\end{equation*}
Then, $\varepsilon(t_k,\bar{\boldsymbol\theta})$ in~\eqref{eq:epsilon} becomes
\begin{align*}
	\varepsilon(t_k,\bar{\boldsymbol\theta}) &= \left\{\frac{B^*(p)}{A^*(p)(1+G^*(p)C(p))}r(t)\right\}_{t=t_k} 
	-\left\{\frac{\bar{B}(p)}{\bar{A}(p)(1+G^*(p)C(p))}r(t)\right\}_{t=t_k} + \varepsilon_u(t_k)\\
	&= \frac{B^*(p)}{A^*(p)(1+G^*(p)C(p))}r(t_k)
	-\frac{\bar{B}(p)}{\bar{A}(p)(1+G^*(p)C(p))}r(t_k) + \varepsilon_u(t_k)	\\
	&= \frac{B^*(p)\bar{A}(p)-\bar{B}(p)A^*(p)}{A^*(p)\bar{A}(p)}\frac{1}{1+G^*(p)C(p)}r(t_k) + \varepsilon_u(t_k)\\
	&= \frac{H(p)}{A^*(p)\bar{A}(p)}\frac{A^*(p)L(p)}{Q^*(p)}r(t_k) + \varepsilon_u(t_k)\\
	&= \frac{L(p)}{\bar{A}(p)Q^*(p)}\mathbf{r}_d^\top(t_k)\mathbf{h} + \varepsilon_u(t_k),
\end{align*}
where $H(p)=B^*(p)\bar{A}(p)-\bar{B}(p)A^*(p)$ and $\mathbf{h}$ is an $(n+m+1)$-th vector containing the 
coefficients of $H(p)$.
	
Recall that the filtered instrument vector is given by~\eqref{eq:phi_hat}. Hence, ~\eqref{eq:epsilon} can be written as
\begin{equation*}
	\mathbf{S}(-\bar{B},\bar{A})\underbrace{\mathbb{E}\left\{\frac{L(p)}{\bar{A}(p)\bar{Q}(p)}\mathbf{r}_d(t_k)
	\frac{L(p)}{\bar{A}(p)\bar{Q}(p)}\mathbf{r}_d^\top(t_k)\right\}}_{=\bar{\boldsymbol\Phi}}\mathbf{h} 
	+ \mathbf{S}(-\bar{B},\bar{A})\underbrace{\mathbb{E}\left\{\frac{L(p)}{\bar{A}(p)\bar{Q}(p)}\mathbf{r}_d(t_k)
	\varepsilon_u(t_k)\right\}}_{=\bar{\boldsymbol\Psi}} = \mathbf{0}.
\end{equation*}
By using the same reasoning as showing the non-singularity of $\boldsymbol\Phi$ in~\eqref{eq:Phi}, we can show that 
$\bar{\boldsymbol\Phi}$ is generically non-singular. 
Note that $\bar{\boldsymbol\Psi}$ is in general not equal to zero due to $\varepsilon_u(t_k) \neq 0$ and is input-dependent.
Since the Sylvester matrix, $\mathbf{S}(-\bar{B},\bar{A})$, is non-singular under Assumption~\ref{A4}, we have
\begin{equation*}
	\mathbf{h} = -\bar{\boldsymbol\Phi}^{-1}\bar{\boldsymbol\Psi} \neq \mathbf{0}.
\end{equation*}
This in turn implies that 
\begin{equation}	\label{eq:clsrivc_sol}
	\frac{\bar{B}(p)}{\bar{A}(p)} = \frac{B^*(p)}{A^*(p)} - \frac{H(p)}{\bar{A}(p)A^*(p)},
\end{equation}
i.e. the unique converging point does not correspond to the true parameter vector.

In addition, the local convergence of the CLSRIVC estimator to the converging point given by~\eqref{eq:clsrivc_sol} 
follows the same procedure as the proof of Statement~2 of Corollary~3 in~\cite{Pan2020}.
Together with $\bar{\boldsymbol\theta}\neq\boldsymbol\theta$, we can conclude that the CLSRIVC 
estimator is not consistent.

\end{proof}

The inconsistency in the CLSRIVC estimator is due the unknown intersample behaviour of the plant input $u(t)$. 
The CLSRIVC estimator can achieve consistency if $u(t)$ is exactly reconstructable. 
A possible modification to the closed-loop system is replacing the continuous-time controller in Figure~\ref{fig:closed-loop1} 
by a discrete-time controller followed by a ZOH or FOH block. 
%Here we are assuming that the controller output, external reference and the output that is fed back to the controller are 
%all discrete-time signals. Only the plant is represented by a continuous-time function.
%Such a closed-loop configuration should therefore fix the interpolation problem 
%that causes $\bar{\boldsymbol\theta} \neq \boldsymbol\theta^*$ in~\eqref{eq:cl_conv_cond}.
%The non-singularity proof might be a little bit tricky to do as we have a combination of discrete-time and continuous-time 
%transfer functions in the sensitivity functions. Still thinking about how to deal with this problem.

%\begin{figure} [H]
%\begin{center}
%\includegraphics[width = 14cm]{closed-loop2}
%\caption{Proposed closed-loop configuration.}
%\label{fig:closed-loop2}
%\end{center}
%\end{figure}

\bibliographystyle{ieeetr}
\bibliography{library}

\end{document}